\documentclass[runningheads]{llncs}
\usepackage{graphicx}
\usepackage{color}
\usepackage{amsmath}
\usepackage{amsfonts}
\usepackage[hidelinks]{hyperref}
\usepackage[T1]{fontenc}
\usepackage[utf8]{inputenc} 
\usepackage{tabularx}
\newbox\ProofSym
\setbox\ProofSym=\hbox{%
   \unitlength=0.18ex%
   \begin{picture}(10,10)
   \put(0,0){\framebox(9,9){}}
   \put(0,3){\framebox(6,6){}}
   \end{picture}}

\newcommand{\ccw}{\textsf{ccw}}
\newcommand{\cw}{\textsf{cw}}

\spnewtheorem{observation}{Observation}{\bfseries}{\rmfamily}

\newcommand{\arc}[1]{\gamma(#1)}

\newif\iffull
\fulltrue

\usepackage{pdfpages}
\usepackage[absolute,overlay]{textpos}

\begin{document}
%
\title{Covering Convex Polygons by Two Congruent Disks\thanks{This research was supported by the Institute of Information \& communications 
Technology Planning \& Evaluation(IITP) grant funded by the Korea government(MSIT)
(No. 2017-0-00905, Software Star Lab (Optimal Data Structure and Algorithmic Applications in Dynamic Geometric Environment)) and (No. 2019-0-01906, Artificial Intelligence Graduate School Program(POSTECH)).}}

%
%
\author{Jongmin Choi\inst{1} \and
Dahye Jeong\inst{1} \and
Hee-Kap Ahn\inst{2}\orcidID{0000-0001-7177-1679}}
\authorrunning{J. Choi et al.}
%

\institute{Department of Computer Science and Engineering, 
Pohang University of Science and Technology, Pohang, Korea\\
\email{\{icothos,dahyejeong\}@postech.ac.kr} \and
Department of Computer Science and Engineering, Graduate School of Artificial Intelligence, 
Pohang University of Science and Technology, Pohang, Korea\\
\email{heekap@postech.ac.kr}}
\maketitle              
\begin{abstract}
We consider the planar two-center problem for a convex polygon: given 
a convex polygon in the plane, find two congruent disks of minimum radius whose union contains 
the polygon. We present an $O(n\log n)$-time
algorithm for the two-center problem for a convex polygon, where $n$ is the number of
vertices of the polygon. This improves upon the previous best algorithm for the problem.

\end{abstract}

\section{Introduction}
The problem of covering a region $R$ by a predefined shape $Q$ (such as a
disk, a square, a rectangle, a convex polygon, etc.) in the plane is
to find $k$ homothets\footnote{For a shape $Q$ in the plane, a (positive) homothet of $Q$ 
is a set of the form $\lambda Q +v := \{\lambda q+v\mid q\in Q\}$,
where $\lambda>0$ is the homothety ratio, and $v\in\mathbb{R}^2$ is a translation vector.}  of $Q$ 
with the same homothety ratio 
such that their union contains $R$ and the homothety ratio is minimized.
The homothets in the covering are allowed to overlap, as long
as their union contains the region. 
This is a fundamental optimization problem~\cite{Agarwal2002,Agarwal1998,Feder1988} arising in analyzing 
and recognizing shapes, 
and it has real-world applications, including computer vision and data mining.

The covering problem has been extensively studied in the context of 
the \emph{$k$-center problem} and the \emph{facility location} problem
when the region to cover is a set of points and the predefined shape is
a disk in the plane.
In last decades, there have been a lot of works, including exact algorithms for $k=2$~\cite{AGAR1994,CHAN1999,CHO2020,CHOI2020,SHARIR1997,WANG2020}, exact and approximation algorithms 
for large $k$~\cite{Agarwal2002,Feder1988,GONZALEZ1985,HWANG1993}, 
algorithms in higher dimensional spaces~\cite{Agarwal2013,Agarwal2002,MEGI1984}, and 
approximation algorithms for streaming points~\cite{Agarwal2015,Ahn2014107,Chan2014,Hershberger2008,KIM2015,Zadeh2008}.
There are also some works on the $k$-center problem for small $k$ when the region
to cover is a set of disks in the plane, for $k=1$~\cite{Fischer2004,Loeffler2010,Megiddo1989} and $k=2$~\cite{AHN2013}.

In the context of the facility location, there have also been some works on 
the \emph{geodesic $k$-center} problem 
for simple polygons~\cite{Ahn2016,OH2018} 
and polygonal domains~\cite{Bae2019}, in which we find 
$k$ points (centers) in order to minimize the maximum geodesic 
distance from any point in the domain to its closest center.



In this paper we consider the covering problem for a convex polygon in which 
we find two congruent disks of minimum radius whose union contains the convex polygon.
Thus, our problem can be considered as the \emph{(geodesic) two-center problem for a convex polygon.} See Fig.~\ref{fig:twoCenter} for an illustration.

\paragraph{Previous works.} For a convex polygon with $n$ vertices, Shin~et~al.~\cite{SHIN1998} gave an 
$O(n^2 \log^3 n)$-time algorithm using parametric search for the two-center
problem. They also gave an $O(n\log^3{n})$-time algorithm
for the restricted case of the two-center problem in which 
the centers must lie at polygon vertices.
Later, Kim and Shin~\cite{KIM2000} improved the results and gave an $O(n \log^3 n \log\log n)$-time algorithm for the two-center problem 
and an $O(n \log^2 n)$-time algorithm for the restricted case of the problem.

There has been a series of work dedicated to variations of the $k$-center problem for a convex polygon, most of which require certain constraints on the centers, including the centers restricted to lie on the polygon boundary~\cite{ROY2008} and
on a given polygon edge(s)~\cite{DAS2008,ROY2008}.
For large $k$, there are quite a few approximation algorithms.
For $k \geq 3$, Das~et~al.~\cite{DAS2008} gave an $(1+\epsilon)$-approximation algorithm with the centers restricted to lie 
on the same polygon edge, along with a heuristic algorithm without such restriction.
Basappa~et~al.~\cite{BASAPPA2020} gave a $(2+\epsilon)$-approximation algorithm for $k \geq 7$, where the centers are restricted to lie on the polygon boundary.
There is a 2-approximation algorithm for the two-center problem for 
a convex polygon that supports insertions and deletions
of points in $O(\log n)$ time per operation~\cite{SADHU2019}.

\begin{figure}%
  \centering
  {\includegraphics[width=.8\textwidth]{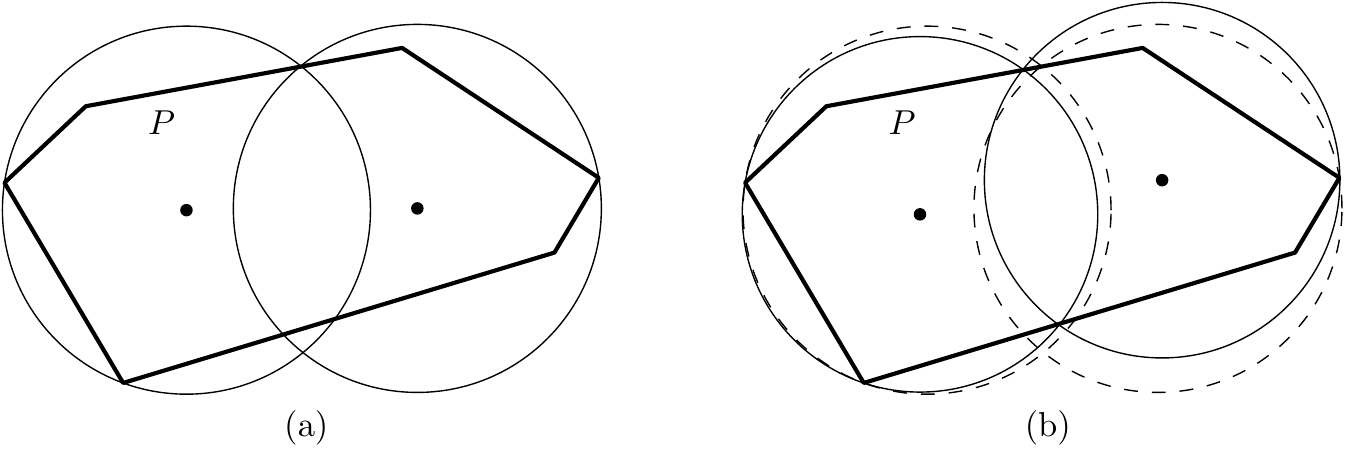}}%
  \caption{(a) Two congruent disks whose union covers a convex polygon $P$. (b) $P$ can be covered by two congruent disks of smaller radius.}%
  \label{fig:twoCenter}%
\end{figure}
   
\paragraph{Our results.}
We present an $O(n\log n)$-time deterministic algorithm for the
two-center problem for a convex polygon $P$ with $n$ vertices. 
That is, given a convex polygon with $n$ vertices, we can find in $O(n\log n)$ time two congruent disks of minimum radius whose union covers the polygon.
This improves upon the $O(n\log^{3}{n}\log{ \log {n}})$
time bound of Kim and Shin~\cite{KIM2000}.
\paragraph{Sketch of our algorithm.}
Our algorithm is twofold.  First we solve the sequential decision
problem in $O(n)$ time. 
That is, given a real value $r$, decide whether $r \geq r^{*}$, where $r^{*}$ is
the optimal radius value.
Then we present a parallel
algorithm for the decision problem which takes $O(\log{n})$
time using $O(n)$ processors, after an $O(n\log{n})$-time preprocessing. 
Using these decision algorithms and applying Cole's
parametric search~\cite{COLE1987}, we solve the optimization problem,
the two centers for $P$, in $O(n\log{n})$ deterministic time.

We observe that if $P$ is covered by two congruent disks 
$D_1$ and $D_2$ of radius $r$, 
$D_1$ covers a connected subchain $P_1$ of the boundary of 
$P$ and $D_2$ covers
the remaining subchain $P_2$ of the boundary of $P$. 
Thus, in the sequential decision algorithm, 
we compute for any point $x$ on the boundary of $P$, the longest subchain of the boundary of $P$ from $x$ in counterclockwise direction that is covered by a disk of radius $r$, 
and the longest subchain of the boundary $P$ from $x$ in clockwise direction that is covered by a disk of radius $r$.
We show that the determinators of the disks that define the two longest subchains 
change $O(n)$ times while $x$ moves along the boundary of $P$.
We also show that the disks and the longest subchains can be represented by $O(n)$ 
algebraic functions.
Our sequential decision algorithm computes the longest subchains in $O(n)$ time.
Finally, the sequential decision algorithm determines whether there is a point $x'$
in $P$ such that the two longest subchains from $x'$, one in counterclockwise direction and one in clockwise direction, cover the polygon boundary in $O(n)$ time.


Our parallel decision algorithm computes the longest subchains 
in parallel and determines whether there is a point $x'$ in $P$ such that 
the two longest subchains from $x'$ covers the polygon boundary in $O(\log n)$ parallel steps using $O(n)$ processors after $O(n\log n)$-time preprocessing.
For this purpose, the algorithm finds rough bounds of the longest subchains,
by modifying the parallel decision algorithm for the planar 
two-center problem of points in convex position~\cite{CHOI2020} 
and applying it for the vertices of $P$.
Then the algorithm computes $O(n)$ algebraic functions of the longest subchains in $O(\log n)$ time using $O(n)$ processors.
Finally, it determines in parallel computation 
whether there is a point $x'$ in $P$ such that 
the two longest subchains from $x$ covers the polygon boundary.


  We can compute the optimal radius value $r^*$ using Cole's parametric search~\cite{COLE1987}.
  For a sequential decision algorithm of running time $T_S$ and 
  a parallel decision algorithm of parallel running time $T_P$ using $N$ processors, 
  Cole's parametric search is a technique that computes an optimal value in $O(NT_P + T_S(T_P + \log N))$ time.
  In our case, $T_S = O(n)$, $T_P = O(\log n)$, and $N = O(n)$.
  Therefore, we get a deterministic $O(n\log n)$-time algorithm for the two-center problem for a convex polygon $P$.

\section{Preliminaries}

  For any two sets $X$ and $Y$ in the plane, we say $X$ \emph{covers} $Y$
  if $Y\subseteq X$. We say a set $X$ is \emph{$r$-coverable} 
  if there is a disk $D$ of radius $r$ covering $X$. For a compact set $A$, 
  we use $\partial A$ to denote the boundary of $A$. We simply say $x$ moves 
  \emph{along} $\partial A$ when $x$ moves in the counterclockwise direction 
  along $\partial A$. Otherwise, we explicitly mention the direction.
 
  Let $P$ be a convex polygon with $n$ vertices $v_1,v_2,\ldots,v_n$ in counterclockwise order 
  along the boundary of $P$.
  Throughout the paper, we assume
  general circular position on the vertices of $P$, meaning no four vertices are cocircular.  
  We denote the subchain of $\partial P$ from a
  point $x$ to a point $y$ in $\partial P$ in counterclockwise
  order as $P_{x,y} = \left\langle x, v_{i}, v_{i+1},\ldots, v_{j}, y  \right\rangle$, where
  $v_{i}, v_{i+1},\ldots, v_{j}$ are the vertices of $P$
  that are contained in the subchain. 
  We call $x, v_{i}, v_{i+1},\ldots, v_{j}, y$ the \emph{vertices} of $P_{x,y}$.
  By $|P_{x,y}|$, we denote the number of distinct vertices of $P_{x,y}$.

  We can define an order on the points of
  $\partial P$, with respect to a point $p \in \partial P$.
  For two points $x$ and $y$ of $\partial P$, we use 
  $x<_{p}y$ if $y$ is farther from $p$ than $x$ in the counterclockwise direction
  along $\partial P$. We define $\leq_{p}, >_{p}, \geq_{p}$ accordingly.

  For a subchain $C$ of $\partial P$, we denote
  by $I_{r}(C)$ the intersection of the disks of radius $r$, 
  each centered at a point in $C$. See Fig.~\ref{fig:intersection}(a). 
  Observe that any disk of radius $r$
  centered at a point $p\in I_{r}(C)$ covers the entire
  chain $C$. Hence, $I_{r}(C) \neq \emptyset$ if and only if $C$ is $r$-coverable.  
  The \emph{circular hull} of a set $X$, denoted by $\alpha_{r}(X)$, 
  is the intersection of all disks of radius $r$ covering $X$.
  See Fig.~\ref{fig:intersection}(b).
  Let $S$ be the set of vertices of a subchain $C$ of $\partial P$.  If a disk
  covers $C$, it also covers $S$.  
  If a disk covers $S$, it covers $C$ since it covers
  every line segment induced by pairs of the points in $S$, due to the convexity
  of a disk.  Therefore, $\alpha_{r}(C)$ and
  $\alpha_{r}(S)$ are the same and $I_{r}(C)$ and $I_r(S)$ are the same. 
\begin{figure}%
  \centering
  {\includegraphics[width=.6\textwidth]{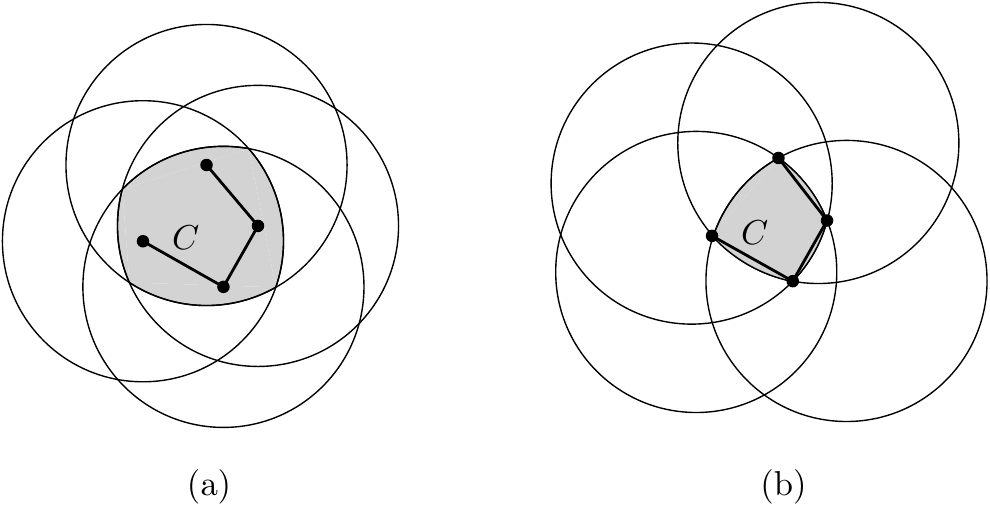}}%
  \caption{$C$ is a subchain of $\partial P$ and $S$ is the vertex set of $C$. (a) $I_r(S)=I_r(C)$ (b) $\alpha_{r}(S)=\alpha_r(C)$}%
  \label{fig:intersection}%
\end{figure}

  Every vertex of $\alpha_{r}(C)$ is a vertex of $C$.  
  The boundary of $\alpha_{r}(C)$ consists of arcs of radius $r$, each connecting 
  two vertices of $C$.
  The circular hull $\alpha_{r}(C)$ is dual to
  the intersection $I_{r}(C)$, in the sense that every arc of
  $\alpha_{r}(C)$ is on the circle of radius $r$ centered at a
  vertex of $I_{r}(C)$, and every arc of $I_{r}(C)$ is on the
  circle of radius $r$ centered at a vertex of $\alpha_{r}(C)$.
  This implies that $\alpha_{r}(C)\neq\emptyset$ if and only if
  $I_{r}(C) \neq \emptyset$.
  Therefore, $\alpha_{r}(C)$ is nonempty if and only if $C$ is $r$-coverable.  
  
  For a vertex $v$ of $\alpha_{r}(C)$, we denote by $\ccw(v)$ its counterclockwise 
  neighbor on $\partial \alpha_{r}(C)$, and by $\cw(v)$ its clockwise 
  neighbor on $\partial \alpha_{r}(C)$.
  We denote by $\arc{v}$ the arc of $\alpha_{r}(C)$ 
  connecting $v$ and $\ccw(v)$ of $\alpha_{r}(C)$.
  By $\delta(v)$, we denote the supporting disk of the arc $\arc{v}$ of $\alpha_{r}(C)$, that is, the disk containing $\arc{v}$ in its boundary.
  We may use $\alpha(C)$ and $I(C)$ to denote $\alpha_r(C)$ and $I_r(C)$, respectively, if it is understood from context. 
Since $\alpha(C)$ and $\alpha(S)$ are the same, we obtain the following observation on 
subchains from the observations on planar points~\cite{EDELS1983,HERSH1991}. 
  \begin{observation}[\cite{EDELS1983,HERSH1991}]\label{obs:alpha} 
  For a subchain $C$ of $\partial P$ the followings hold.
    \begin{enumerate}
      \item For any subchain $C' \subseteq C$, $\alpha_{r}(C') \subseteq \alpha_{r}(C)$~.
      \item A vertex of $C$ appears as a vertex in $\alpha_{r}(C)$ if and only if 
      $C$ is $r$-coverable by a disk containing the vertex on its boundary.
      \item An arc of radius $r$ connecting two vertices of $C$ appears as an arc of
      $\alpha_{r}(C)$ if and only if $C$ is $r$-coverable by the supporting disk of the arc.
    \end{enumerate}
  \end{observation}

  \noindent
  For a point $x\in \partial P$, let $f_{r}(x)$ be the farthest point on $\partial P$ from $x$ in the counterclockwise direction along $\partial P$ such that $P_{x,f_r(x)}$ is $r$-coverable.
  We denote by $D_1^{r}(x)$ the disk of radius $r$ covering $P_{x,f_r(x)}$.
  Similarly, let $g_{r}(x)$ be the farthest point on $\partial P$ from $x$ in the clockwise direction 
  such that $P_{g_r(x),x}$ is $r$-coverable,
  and denote by $D_2^{r}(x)$ the disk of radius $r$ covering $P_{g_r(x),x}$.
  Note that $x$ may not lie on the boundaries of $D_{1}^{r}$ and $D_{2}^{r}$.
  We may use $f(x)$, $D_{1}(x)$, $g(x)$, and $D_{2}(x)$ by omitting
  the subscript and superscript $r$ in the notations, if they are understood from context.

  Since we can determine in $O(n)$ time whether $P$ is $r$-coverable~\cite{MEGI1984},
  we assume that $P$ is not $r$-coverable in the remainder of the paper.
  For a fixed $r$, consider any two points $t$ and $t'$ in $\partial P$ 
  satisfying $t<_{t} t'<_{t} f(t)$.
  Then $P_{t',f(t)}$ is $r$-coverable, which implies
  $f(t)\leq_{t'} f(t')$. Thus, we have the following observation.

  \begin{observation} \label{obs:monotone}
    For a fixed $r$, as $x$ moves along $\partial P$ in the counterclockwise direction, 
    both $f(x)$ and $g(x)$ move monotonically along $\partial P$ in the counterclockwise direction.
  \end{observation}
  
  \section{Sequential Decision Algorithm}
  \label{sec:sequential}
  In this section, we consider the decision problem: given a real value $r$, decide whether $r\geq r^{*}$,
  that is, whether there are two congruent disks of radius $r$ whose union covers $P$.
  
  For a point $x$ moving along $\partial P$, we consider 
  two functions, $f(x)$ and $g(x)$.
  If there is a point $x\in \partial P$ such that $f(x) \geq_{x} g(x)$, the union of 
  $P_{x,f(x)}$ and $P_{g(x),x}$ is $\partial P$. Thus there are 
  two congruent disks of radius $r$ whose union covers $P$, and
  the decision algorithm returns \texttt{yes}. 
  Otherwise, we conclude that $r < r^{*}$, and the decision algorithm returns \texttt{no}.
  For a subchain $P_{x,y}$ of $\partial P$, we use $\alpha(x,y)$ to denote $\alpha(P_{x,y})$,
  and $I(x,y)$ to denote $I(P_{x,y})$.

\subsection{Characterizations} \label{sec:sequential.characterization}
  For a fixed $r$, $I(x,f(x))$ is a point, and it is the center of $\alpha(x,f(x))$.
  Moreover, $\alpha(x,f(x))$ and $D_1(x)$ are the same.   
  Observe that $D_1(x)$ is defined by two or three vertices of $P_{x,f(x)}$, 
  which we call the \emph{determinators} of $D_1(x)$.
  For our purpose, we define four \emph{types} of $D_1(x)$ by its determinators:
  (\texttt{T1}) $x$, $f(x)$, and one vertex. (\texttt{T2}) $x$ and $f(x)$.
(\texttt{T3}) $f(x)$ and one vertex. (\texttt{T4}) $f(x)$ and two vertices. 
  See Fig.~\ref{fig:types} for an illustration of the four types.

  \begin{figure}[ht]%
    \centering
    {\includegraphics[width=\textwidth]{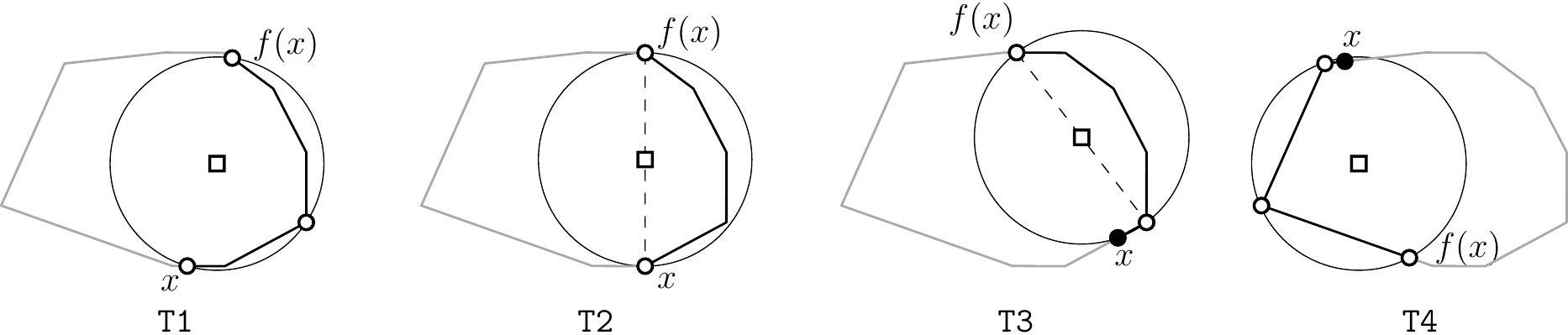}}%
    \caption{Four types of $D_1(x)$ and its determinators (small circles).}
    \label{fig:types}%
  \end{figure}

  We denote by $e(a)$ the edge of $P$ containing a point $a\in\partial P$. 
  If $a$ is a vertex of $P$, $e(a)$ denotes the edge of $P$ incident to $a$ 
  lying in the counterclockwise direction from $a$.
  For a point $x$ moving along $\partial P$, the \emph{combinatorial structure} of $f(x)$ is determined by $e(x)$, $e(f(x))$, and the determinators of $D_1(x)$. 
  We call each point $x$ in $\partial P$ at which the combinatorial structure of $f(x)$ changes 
  a \emph{breakpoint} of $f(x)$.
  For $x\in \partial P$ lying in between two consecutive breakpoints, we can compute $f(x)$ using $e(x)$, $e(f(x))$, and $D_1(x)$.

  Consider $x$ moving along $\partial P$ starting from $x_0$ on $\partial P$ 
  in counterclockwise direction. Let $x_{1} = f(x_{0})$, $x_{2} = f(x_{1})$ and $x_3=f(x_2)$.
  We simply use the index $i$ instead of $x_i$ for $i=0,\ldots,3$ if it is understood from context. 
  For instance, we use $P_{i,j}$ to denote $P_{x_i,x_j}$, and $\leq_{i}$ to denote $\leq_{x_i}$.
  For the rest of the section, we describe how to handle the case that $x$ moves along 
  $P_{0,1}$. The cases that $x$ moves along $P_{1,2}$ and $P_{2,3}$ can be handle analogously.
  As $x$ moves along $P_{0,1}$, $f(x)$ moves along $P_{1,2}$ in the same direction by
  Observation~\ref{obs:monotone}. 

  \begin{lemma}\label{lem:4chains}
  For any fixed $r\geq r^*$, the union of $P_{0,1}$, $P_{1,2}$, and $P_{2,3}$ is $\partial P$.
  \end{lemma}
  \iffull 
  \begin{proof}
    If $P$ is $r$-coverable, $P_{0,1}$ is $\partial P$. Assume that $P$ is not $r$-coverable. 
    For any fixed $r\geq r^*$, there are two congruent disks $D_1$ and $D_2$  of radius $r$ 
    whose union covers $P$. 
    Let $y$ and $z$ be the points of $\partial P$ such that
    $P_{y,z}$ is covered by $D_1$, and $P_{z,y}$ is covered by $D_2$.
    Without loss of generality, assume $x_0\in P_{y,z}$. Then $z \leq_0 f(y) \leq_0 x_1$ along $\partial P$, 
    because $P_{y,z}$ is covered by $D_1$ and by Observation~\ref{obs:monotone}.
    If $x_1\leq_0 y$, then $y\leq_1 x_2$ and thus $x_0 \leq_y x_3$. 
    If $y\leq_0 x_1$, then $x_0\leq_y x_2$.
    Thus, the union of $P_{0,1}$, $P_{1,2}$, and $P_{2,3}$ is $\partial P$.
  \end{proof}
  \fi
  
  The structure of a circular hull can be expressed by 
  the circular sequence of arcs appearing on the boundary of the circular hull.
  There is a 1-to-1 correspondence between a breakpoint of $f(x)$ for $x$ moving along $P_{0,1}$ 
  and a structural change to $\alpha(x,f(x))$. 
  This is because $D_1(x)$ and $\alpha(x,f(x))$ are the same.
  Thus, we maintain $D_1(x)$ for $x$ moving along $P_{0,1}$ and capture every structural change to 
  $\alpha(x,f(x))$. Observe that the boundary of $\alpha(x,f(x))$ consists of a connected boundary part of 
  $\alpha(x,x_1)$, a connected boundary part of $\alpha(x_1,f(x))$, and two arcs of $D_1(x)$ 
  connecting $\alpha(x,x_1)$ and $\alpha(x_1,f(x))$. See Fig.~\ref{fig:type1} for an illustration.
  
  \begin{figure}[ht]%
    \centering
    {\includegraphics[width=.9\textwidth]{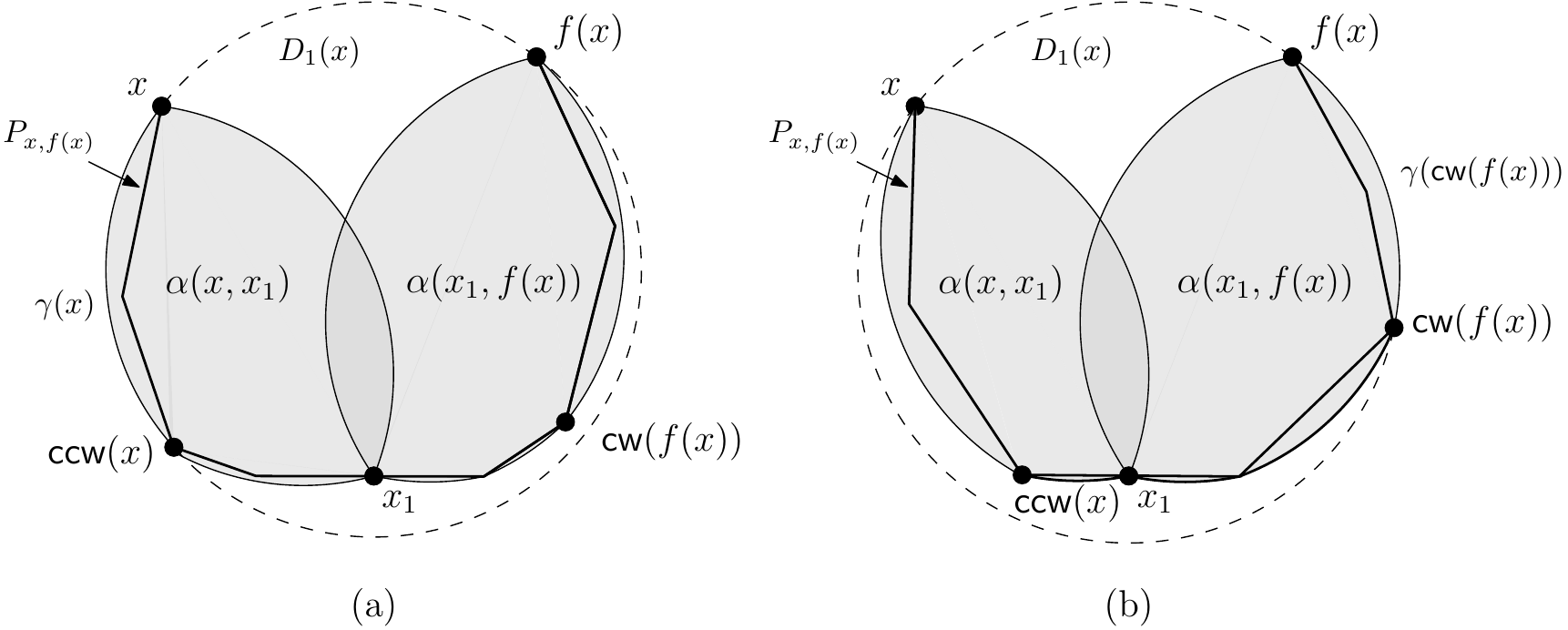}}%
    \caption{Two cases of $D_1(x)$ of type \texttt{T1}. Two arcs (dashed) of $D_1(x)$ connecting $\alpha(x,x_1)$
    and $\alpha(x_1,f(x))$.
    (a) If $v$ is on the boundary of $\alpha(x,x_1)$, $D_1(x)$ is
    $\delta(x)$ of $\alpha(x,x_1)$. (b) If $v$ is on the boundary of $\alpha(x_1,f(x))$, $D_1(x)$ is
    $\delta(\textsf{cw}(f(x)))$.}
    \label{fig:type1}%
  \end{figure}
  
  The following lemmas give some characterizations to the four types of $D_1(x)$.
  Recall that $\delta(v)$ is the supporting disk of the arc $\arc{v}$ of an circular hull, that is, the disk containing $\arc{v}$ on its boundary.
  \iffull
  \begin{lemma} \label{lem:type1}
    Any disk $D_1(x)$ of type \texttt{T1} is 
    $\delta(x)$ of $\alpha(x,x_1)$ or $\delta(\cw(f(x)))$ of $\alpha(x_1,f(x))$. 
  \end{lemma}
  \begin{proof}
  Since the determinators of $D_1(x)$ are $x$, $f(x)$, and a vertex $v$ of $P_{x,f(x)}$,
  they all appear on the boundary of $\alpha(x,f(x))$.
  Assume that $v$ is on the boundary of $\alpha(x,x_1)$. Then the boundary portion of $\alpha(x,f(x))$,
  from $x$ to $v$ in counterclockwise order, is from the boundary of $\alpha(x,x_1)$, and
  $\ccw(x)$ of $\alpha(x,f(x))$ lies on the boundary of $\alpha(x,x_1)$.
  Since $\arc{x}$ of $\alpha(x,x_1)$ is on the boundary of $D_1(x)$, 
  $\delta(x)$ of $\alpha(x,x_1)$ is $D_1(x)$. See Fig.~\ref{fig:type1}(a).
  A similar argument can be made for the case that $v$ is on the boundary of 
  $\alpha(x_1,f(x))$. See Fig.~\ref{fig:type1}(b). 
  Therefore, $D_1(x)$ is $\delta(x)$ of $\alpha(x,x_1)$ or $\delta(\cw(f(x)))$ of $\alpha(x_1,f(x))$. 
  \end{proof}

\begin{lemma} \label{lem:type2}
For a disk $D_1(x)$ of type \texttt{T2}, the Euclidean distance between $x$ and $f(x)$ is $2r$.
\end{lemma}    
\begin{proof}
Since $x$ and $f(x)$ are the determinators of $D_1(x)$, their Euclidean distance is $2r$.
\end{proof}

\begin{lemma} \label{lem:types34}
   If a disk $D_1(x)$ of type \texttt{T3} or \texttt{T4} has $x$ on its boundary, 
   $D_1(x)$ is $\delta(x)$ of $\alpha(x,x_1)$ or $\delta(\cw(f(x)))$ of $\alpha(x_1,f(x))$. 
   Moreover, for any point $y$ in the interior of $P_{x,v}$,
   $D_1(y)$ has the same type as $D_1(x)$, where $v$ is the determinator of
   $D_1(x)$ closest to $x$ in counterclockwise order.
\end{lemma}
\begin{proof}
By an argument similar to the proof of Lemma~\ref{lem:type1}, we can show that 
$D_1(x)$ is $\delta(x)$ of $\alpha(x,x_1)$ or $\delta(\cw(f(x)))$ of $\alpha(x_1,f(x))$.
By Observation~\ref{obs:monotone}, we have $f(x)\leq_v f(y)$.
Thus the determinators of $D_1(x)$ lying in between $v$ and $f(x)$ are all contained in $P_{y,f(y)}.$
Therefore, $D_1(x)$ and $D_1(y)$ are the same.
\end{proof}
\fi
If there is a change
to $e(x)$, $e(f(x))$, $\ccw(x)$ of $\alpha(x,x_1)$ or $\cw(f(x))$ of $\alpha(x_1,f(x))$, the combinatorial structure of $f(x)$ changes.
Therefore, we compute the changes to $e(x)$ and $\ccw(x)$ of $\alpha(x,x_1)$ for a point 
$x$ moving along $P_{0,1}$, and compute the changes to $e(y)$ and $\cw(y)$ of $\alpha(x_1,y)$ for 
a point $y$ moving along $P_{1,2}$.
We call the points inducing these changes the \emph{event points}.
From this, we detect the combinatorial changes to $f(x)$.

\subsection{Data structures and decision algorithm} \label{sec:sequential.datastructures}
Wang~\cite{Wang2020full} proposed a semi-dynamic (insertion-only) data structure for maintaining the circular hull for points in the plane that are inserted in increasing order of 
their $x$-coordinates.
It is also mentioned that the algorithm can be modified to work for points that are inserted 
in the sorted order around a point. Since the vertices of $P$ are already sorted around 
any point in the interior of $P$, we can use the algorithm for our purpose.

  \begin{lemma}[\protect{Theorem 5 in \cite{Wang2020full}}]\label{lem:Wang}
    We can maintain the circular hull of a set $Q$ of points such that when a new point to the right of all points of $Q$ is inserted, we can decide in $O(1)$ amortized time whether $\alpha(Q)$ is nonempty, and  update $\alpha(Q)$.
  \end{lemma}
  
  We can modify the algorithm to work not only for point insertions, but also for edge insertions.
  Let $v_1,\ldots, v_i$ be the vertices of $P$ inserted so far in order from $v_1$. When $v_{i+1}$ is inserted, 
  we compute the points $z$ on edge $v_iv_{i+1}$ at which a structural change to $\alpha(v_1,z)$ occurs.

\begin{lemma} \label{lem:precomputeEvent1}
 For a point $x$ moving along $P_{0,1}$, $e(x)$ and $\ccw(x)$ of $\alpha(x,x_1)$ change 
 $O(|P_{0,1}|)$ times.
 We can compute the event points $x$ at which $\ccw(x)$ of $\alpha(x,x_1)$ changes 
 in $O(|P_{0,1}|)$ time.
\end{lemma}
\iffull
\begin{proof}
  Imagine that $x$ moves along the subchain $P_{0,1}$ in \emph{clockwise} order, 
  from $x_1$ to $x_0$. We consider $e(x)$ and $\ccw(x)$ of $\alpha(x,x_1)$ 
  while inserting the vertices of $P_{0,1}= \left\langle x_0,\ldots, v_{i}, v_{i+1},\ldots, x_1  \right\rangle$
  in reverse order, one by one from $x_1,$ and compute the circular hull of the inserted vertices. 
  Then $e(x)$ changes $O(|P_{0,1}|)$ times at the vertices of $P_{0,1}$.
  
  When $v_i$ is inserted, we have $\alpha(v_{i+1},x_1)$.
  We also have the vertices of the hull stored in a stack in clockwise order with $v_{i+1}$ at the top.
  Our goal is to compute $\alpha(v_i,x_1)$ and to compute the points $x$ on ${v_{i}v_{i+1}}$ 
  at which $\ccw(x)$ of $\alpha(x,x_1)$ changes. To do this,
  we pop vertices repeatedly from the stack until $\ccw(v_{i})$ of $\alpha(v_i,x_1)$ 
  becomes the top element of the stack. When a vertex $v$ is popped from the stack, 
  we consider the supporting disk $D$ of the arc connecting $v$ and 
  the vertex at the top of the stack at the moment,
  and compute the intersection of $\partial D$ with the edge ${v_{i}v_{i+1}}$.
  Since $\ccw(x)$ of $\alpha(x,x_1)$ changes when $x$ reaches such an intersection,
  we consider those intersection points $x$ as the event points inducing 
  the changes to $\ccw(x)$ of $\alpha(x,x_1)$.
    
  Observe that once a vertex is popped from the stack, it is never inserted to the stack again.
  Therefore, $\ccw(x)$ of $\alpha(x,x_1)$ changes $O(| P_{0,1}|)$ times while $x$ moves 
  along $P_{0,1}$, and we can compute the event points in $O(| P_{0,1}|)$ time.
\end{proof}
\fi
  From Lemma~\ref{lem:precomputeEvent1}, we obtain the following Corollary.

 \begin{corollary} \label{cor:precomputeEvent2}
 For a point $y$ moving along $P_{1,2}$,  
 $e(y)$ and $\cw(y)$ of $\alpha(x_1,y)$ change $O(|P_{1,2}|)$ times. We can compute 
 the event points $y$ at which $\cw(y)$ of $\alpha(x_1,y)$ changes in $O(|P_{1,2}|)$ time.
 \end{corollary}
 
  The event points subdivide $P_{0,1}$ and $P_{1,2}$ into 
  $O(|P_{0,1}|)$ and $O(|P_{1,2}|)$ pieces, respectively. 
  Since the vertices of $P_{0,1}$ and $P_{1,2}$ 
  are also event points (defined by the changes to $e(x)$ and $e(y)$), 
  each piece is a segment contained in an edge. Moreover,
  any point $x$ in a segment of $P_{0,1}$ has the same $\ccw(x)$ of $\alpha(x,x_1)$, and
  any point $y$ in a segment of $P_{1,2}$ has the same $\cw(y)$ of $\alpha(x_1,y)$.
  
  \iffull
  \begin{lemma} \label{lem:complexity}
   For a fixed $r\geq r^*$, there are $O(n)$ breakpoints of $f(x)$ and $g(x)$.
  \end{lemma}
  \begin{proof}
    For a point $x$ moving along $P_{0,1}$, $e(x)$ and $\ccw(x)$ of $\alpha(x,x_1)$ change $O(|P_{0,1}|)$ times by Lemma~\ref{lem:precomputeEvent1}.
    Since $f(x)$ moves along $P_{1,2}$ while $x$ moves along $P_{0,1}$,
    $e(f(x))$ and $\cw(f(x))$ of $\alpha(x_1,f(x))$ change $O(|P_{1,2}|)$ times
    by Corollary~\ref{cor:precomputeEvent2}.
    Combining them, there are $O(|P_{0,2}|)$ changes to $e(x)$, $e(f(x))$, $\ccw(x)$ 
    of $\alpha(x,x_1)$, and $\cw(f(x))$ of $\alpha(x_1,f(x))$ while $x$ moves along $P_{0,1}$.
    For a fixed $r\geq r^*$, $\partial P$ is covered by the union of at most three subchains 
    $P_{0,1}, P_{1,2}$, and $P_{2,3}$ by Lemma~\ref{lem:4chains}, 
    and thus the total number of changes to $e(x)$, $e(f(x))$, $\ccw(x)$, and $\cw(f(x))$ is $O(n)$.

    Let $T$ be a segment contained in an edge of $P_{0,1}$ 
    such that $e(x)$, $e(f(x))$, $\ccw(x)$ of $\alpha(x,x_1)$, 
    and $\cw(f(x))$ of $\alpha(x_1,f(x))$ remain the same for any $x\in T$.
    We count the breakpoints of $f(x)$ in the interior of $T$. 
    We count the breakpoints of $f(x)$ in the following three cases:
    (1) the type of $D_1(x)$ changes to \texttt{T3} or \texttt{T4}, (2) while $D_1(x)$ is of type \texttt{T1}, 
    the determinators of $D_1(x)$ changes, (3) the type of $D_1(x)$ switches between \texttt{T1} and \texttt{T2}.

    Consider the breakpoints induced by case (1). Once the type of $D_1(x)$ changes to \texttt{T3} or \texttt{T4}, the determinators remain the same while $x$ moves along $T$ by Lemma~\ref{lem:types34}.
    Thus, there is at most one breakpoint induced by case (1).
            
    A breakpoint induced by case (2) is a moment $x=t$, where $\delta(t)$ of $\alpha(t,x_1)$ and $\delta(\cw(f(t)))$ of $\alpha(x_1,f(t))$ 
    are the same by Lemma~\ref{lem:type1}. 
    For $\ccw(x)$ of $\alpha(x,x_1)$ and $\cw(x)$ of $\alpha(x_1,x)$, there is at most one such disk.
    Thus, there is at most one breakpoint induced by case (2).

    Finally, we count the breakpoints of case (3).
    A breakpoint induced by case (3) is a moment $x=t$ where $f(t)$ is at Euclidean distance 
    $2r$ from $t$. Assume that $D_1(x)$ is $\delta(x)$ of $\alpha(x,x_1)$. 
    Then there are at most two such moments $t$ 
    with $\ccw(t)$ on $\partial D_1(t)$
    at which the Euclidean distance between $t$ and $f(t)$ is $2r$.
    For the case that $D_1(x)$ is $\delta(\cw(f(x)))$ of $\alpha(x_1,f(x))$, there are at most two such moments by a similar argument.
    Thus, there are $O(1)$ breakpoints induced by case (3).

    Since there are $O(n)$ segments such that $e(x)$, $e(f(x))$, $\ccw(x)$ of $\alpha(x,x_1)$, 
    and $\cw(f(x))$ of $\alpha(x_1,f(x))$ remain the same for any $x$ in each segment, 
    and there are at most $O(1)$ breakpoints in the interior of 
    a segment, we conclude that there are $O(n)$ breakpoints in total.
    We can show that there are $O(n)$ breakpoints of $g(x)$ by a similar argument.
  \end{proof}

    \begin{figure}[ht]%
    \centering
    {\includegraphics[width=\textwidth]{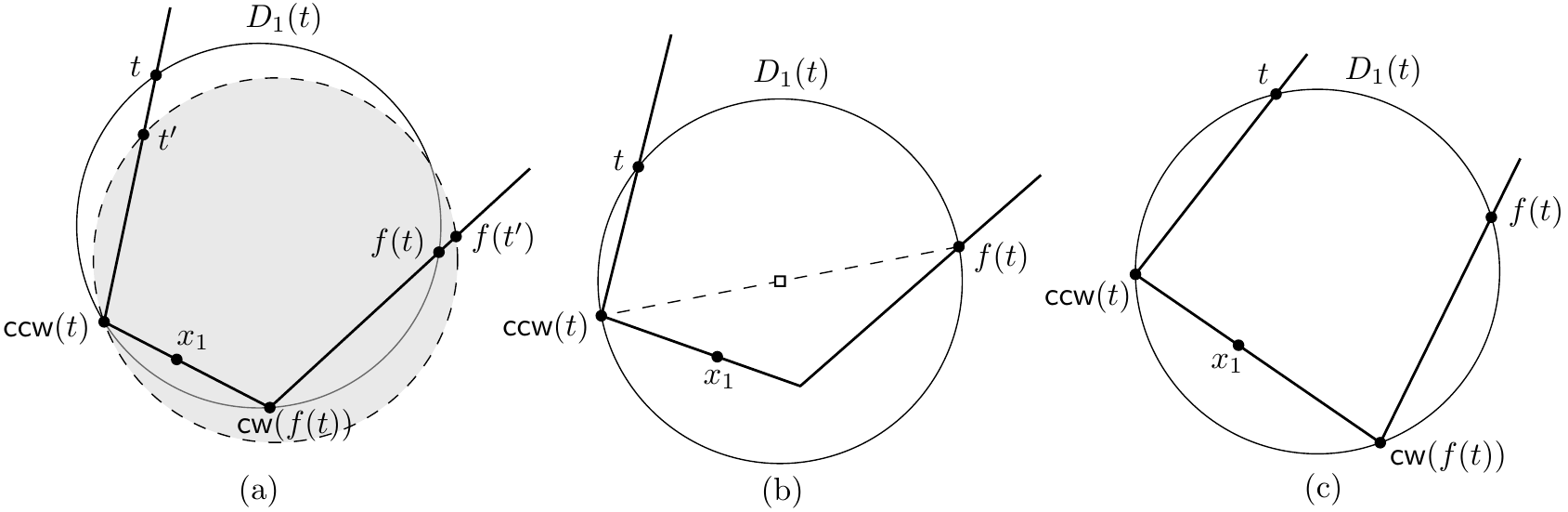}}%
    \caption{Changes to the determinators of $D_1(x)$. 
    (a) At $x=t$, the determinators of $D_1(x)$ of type $\texttt{T1}$ changes from $x,\ccw(x),f(x)$ to $x,\cw(f(x)),f(x)$. 
    (b) At $x=t$, the type of $D_1(x)$ changes to $\texttt{T3}$.
    (c) At $x=t$, the type of $D_1(x)$ changes to $\texttt{T4}$.    }
    \label{fig:typeChange}%
  \end{figure}

  \begin{lemma} \label{lem:SqCompute}
    For a fixed $r\geq r^*$, the breakpoints of $f(x)$ and
    $g(x)$ can be computed in $O(n)$ time.
  \end{lemma}
  \begin{proof}
    We first compute the segments of $P_{0,1}$ induced by the event points of $e(x)$ and 
    $\ccw(x)$ of $\alpha(x,x_1)$,
    and the segments of $P_{1,2}$ induced by the event points of $e(y)$ and $\cw(y)$ of $\alpha(x_1,y)$
    in $O(|P_{0,2}|)$ time by Lemma~\ref{lem:precomputeEvent1}.
    Then we compute the breakpoints of $f(x)$ for points $x\in P_{0,1}$.

    We compute $D_1(x_0)$ and its determinators. 
    To this end, we find the edge $vv'$ of $P$ containing $f(x_0)$ in $O(n)$ time using Lemma~\ref{lem:Wang}. More precisely, 
    if $D_1(x_0)$ is of type \texttt{T1} or \texttt{T2}, $f(x_0)$ is the event point
    at which $\alpha(x_0,f(x_0))$ becomes a disk of radius $r$, and thus we can find 
    $D_1(x_0)$ and its determinators in $O(n)$ time. 
    If $D_1(x_0)$ is of type \texttt{T3} or \texttt{T4}, 
    $x_0$ may not lie on $\partial D_1(x_0)$.
    If $D_1(x_0)$ is of type \texttt{T3}, we find the point $z$ on $vv'$ for every vertex $u$
    of $\alpha(x_0,v)$
    such that the Euclidean distance between $z$ and $u$ is $2r$. 
    In this way, we find $f(x_0)$ in $O(n)$ time for the case. 
    If $D_1(x_0)$ is of type \texttt{T4}, we find $f(x_0)$ in $O(n)$ time
    by computing the intersection of $vv'$ and the supporting disk of every arc of 
    $\alpha(x_0,v)$.
    Therefore, in $O(n)$ time, we can compute $f(x_0)$ and the type of $D_1(x_0)$.
        
    We then continue to compute $f(x)$ for each type of $D_1(x)$ by checking the breakpoints of $f(x)$
    for $x$ moving along $P_{0,1}$ starting from $x_0$.
    Imagine that $x$ moves along a segment $ab$ such that 
    $e(x)$ and $\ccw(x)$ of $\alpha(x,x_1)$ remain the same, 
    and $f(x)$ moves along a segment $cd$ such that 
    $e(f(x))$ and $\cw(f(x))$ of $\alpha(x_1,f(x))$ remain the same.
    If the type of $D_1(x)$ changes to type \texttt{T1}, \texttt{T3}, or \texttt{T4}, 
    $D_1(x)$ becomes $\delta(x)$ of $\alpha(x,x_1)$ or $\delta(\cw(f(x)))$ of $\alpha(x_1,f(x))$ 
    by Lemma~\ref{lem:type1} and~\ref{lem:types34}. 
    At that moment, the determinators of $D_1(x)$ must lie on the boundary of $\alpha(x,f(x))$. 
    By the general circular position assumption, there can be at most three polygon vertices
    lying on the boundary of $D_1(x)$. Therefore, the determinators of $D_1(x)$ are
    at most three elements among the vertices 
    $\ccw(x)$, $\ccw(\ccw(x))$, $\ccw(\ccw(\ccw(x)))$ of $\alpha(x,x_1)$, 
    $\cw(f(x))$, $\cw(\cw(f(x)))$, and $\cw(\cw(\cw(f(x))))$ of $\alpha(x_1,f(x))$. 
    The number of changes to these vertices is $O(n)$ by an argument similar to 
    Lemma~\ref{lem:precomputeEvent1}. 
    Thus, we can compute $f(x)$ and disk $D_1(x)$ under the assumption that $D_1(x)$ is of type \texttt{T1}, \texttt{T3}, and \texttt{T4} in $O(n)$ time.
    See Fig.~\ref{fig:typeChange} for examples.
    If $D_1(x)$ is of type \texttt{T2}, then the Euclidean distance between $x$ and $f(x)$ is $2r$ by Lemma~\ref{lem:type2}.

    By comparing two $f(x)$ functions, one for the current type of $D_1(x)$ and one for other types of $D_1(x)$, we can compute in $O(1)$ time the next breakpoint and the new determinators and type of $D_1(x)$ for $x$ right after the breakpoint.
    Thus, all breakpoints of $f(x)$ can be computed in $O(n)$ time.
    Similarly, the breakpoints of $g(x)$ can be computed in $O(n)$ time.
  \end{proof}

\fi

  Recall that our algorithm returns \texttt{yes} if there exists a point $x \in \partial P$ such that $f(x)\geq_{x} g(x)$, otherwise it returns \texttt{no}. Hence, using Lemma~\ref{lem:SqCompute}, we have the following theorem.

  \begin{theorem}\label{thm:sequential}
    Given a convex polygon $P$ with $n$ vertices in the plane and a radius $r$, 
    we can decide whether there are two congruent disks of radius $r$ covering $P$ in $O(n)$ time.
  \end{theorem}
  \section{Parallel Decision Algorithm}
  Given a real value $r$, our parallel decision algorithm computes $f(x)$ and $g(x)$ that define
  the longest subchains of $\partial P$ from $x$ covered by disks of radius $r$, and 
  determines whether there is a point $x\in\partial P$ such that $f(x) \geq_x g(x)$, in parallel.
  To do this efficiently, our algorithm first finds rough bounds of $f(x)$ and $g(x)$ 
  by modifying the parallel decision algorithm for the two-center problem for points in convex position 
  by Choi and Ahn~\cite{CHOI2020} and applying it for the vertices of $P$. 
  Then our algorithm computes $f(x)$ and $g(x)$ exactly.
  
  The parallel decision algorithm by Choi and Ahn runs in two phases: the preprocessing phase and the decision phase. 
  In the preprocessing phase, their algorithm runs sequentially without knowing $r$. 
  In the decision phase, their algorithm runs in parallel for a given value $r$. 
  It constructs a data structure that supports intersection queries of 
  a subset of disks centered at input points
  in $O(\log n)$ parallel time using $O(n)$ processors after $O(n\log n)$-time preprocessing. 
  In our problem, two congruent disks must cover 
  the edges of $P$ as well as the vertices of $P$, and thus we modify the preprocessing phase.

  In the preprocessing phase, their algorithm partitions the vertices 
  of $P$ into two subsets $S_1=\{v_1,\ldots,v_{k}\}$ and $S_2=\{v_{k+1},\ldots,v_{n}\}$, 
  each consisting of consecutive vertices along $\partial P$ such that there are 
  $v_i \in S_1$ and $v_j\in S_2$ satisfying $\{v_i,v_{i+1},\ldots,v_{j-1}\} \subset D_1$ and 
  $\{v_j,v_{j+1},\ldots,v_{i-1}\}\subset D_2$ for an optimal pair $(D_1,D_2)$ of disks for the vertices of $P$.
  The indices of vertices are cyclic such that $n+k \equiv k$ for any integer $k$.
  Then in $O(n\log n)$ time, it finds $O(n/\log^6 n )$ pairs of subsets, each consisting of 
  $O(\log^6 n)$ consecutive vertices such that there is one pair $(U,W)$ of sets 
  with $v_i\in U$ and $v_j\in W$, where $v_i$ and $v_j$ are the vertices that determine the optimal partition.

  In the preprocessing phase, our algorithm partitions 
  $\partial P$ into two subchains. Then, we partition $\partial P$ into $O(n/\log^6n)$
  subchains, each consisting of $O(\log^6 n)$ consecutive vertices, and
  compute $O(n/\log^6 n )$ pairs of the subchains such that
  at least one pair has $x$ in one subchain and $x'$ in the other subchain, and 
  $P_{x,x'}$ and $P_{x',x}$ is $r^*$-coverable.

  In the decision phase, their algorithm constructs a data structure in $O(\log n)$ parallel time with $O(n)$ processors, 
  that for a query with $r$ computes $I_r(u,w)$, 
  where $u \in U', w \in W'$ for any pair $(U',W')$ among the $O(n/\log^6 n )$ pairs.
  Then it computes $I(u,w)$ in $O(\log n)$ time and determines
  if $I(u,w)=\emptyset$ in $O(\log^3 \log n)$ time using the data structure. 
  
  In our case, our algorithm constructs a data structure that for a query with $r$ computes 
  $I_r(v_i,v_j)$ and $I_r(v_j,v_i)$ for $v_i\in P_1$, $v_j\in P_2$, 
  where $(P_1,P_2)$ is one of the $O(n/\log^6 n)$ pairs of subchains 
  computed in our preprocessing phase. Our data structure also determines if $I(v_i,v_j) = \emptyset$.

  Using the data structure, our algorithm gets rough bounds of $f(x)$ and $g(x)$.         
  Then it computes $f(x)$ and $g(x)$ exactly. 
  In doing so, it computes all breakpoints of $f(x)$ and $g(x)$, 
  and their corresponding combinatorial structures, and determines 
  whether there exists $x \in \partial P$ such that $f(x) \geq_x g(x)$.

  \subsection{Preprocessing phase}
  \iffull
  We use $f^{*}(x)$ and $g^{*}(x)$ to denote $f_{r^*}(x)$ and $g_{r^*}(x)$, respectively.
  Our algorithm partitions $\partial P$ into two subchains such that 
  $P_{x,x'}$ and $P_{x',x}$ are $r^*$-coverable, for $x$ and $x'$ contained in each subchain.
  Then it computes $O(n/\log^6 n)$ pairs of subchains of $\partial P$, each consisting of $O(\log^6 n)$ consecutive vertices.
  
  To this end, our algorithm computes a step function $F(x)$ approximating $f^*(x)$ and a step function 
  $G(x)$ approximating $g^*(x)$ on the same set of intervals of the same length.
  More precisely, at every $(\log^{6}{n})$-th vertex $v$ from $v_1$ along $\partial P$, 
  it evaluates step functions $F(v)$ and $G(v)$ on $r^{*}$  
  such that $f^*(v)\leq_v F(v)$ and $g^*(v)\geq_v G(v)$. See Fig.~\ref{fig:subcell}(a).
  In each interval, the region bounded by $F(x)$ from above and by $G(x)$ from below 
  is a rectangular cell. Thus, there is a sequence of $O({n / \log^{6}{n}})$ rectangular cells of width 
  at most $\log^{6}{n}$. See Fig.~\ref{fig:subcell}(b).
  Observe that every intersection of $f^*(x)$ and $g^*(x)$ is contained in 
  one of the rectangular cells. Thus, we focus on the sequence of rectangular cells bounded in between
  $F(x)$ and $G(x)$, which we call  the \emph{region of interest} (ROI shortly).
  In addition, we require $F(x)$ and $G(x)$ to approximate $f^*(x)$ and $g^{*}(x)$ tight enough such
  that each rectangular cell can be partitioned further by horizontal lines into 
  disjoint rectangular cells of height at most $\log^{6}{n}$, and in total there are $O(n/\log^{6}{n})$ disjoint 
  rectangular cells of width and height at most $\log^{6}{n}$ in ROI. See Fig.~\ref{fig:subcell}(c).

    \begin{figure}[ht]%
      \centering
      {\includegraphics[width=\textwidth]{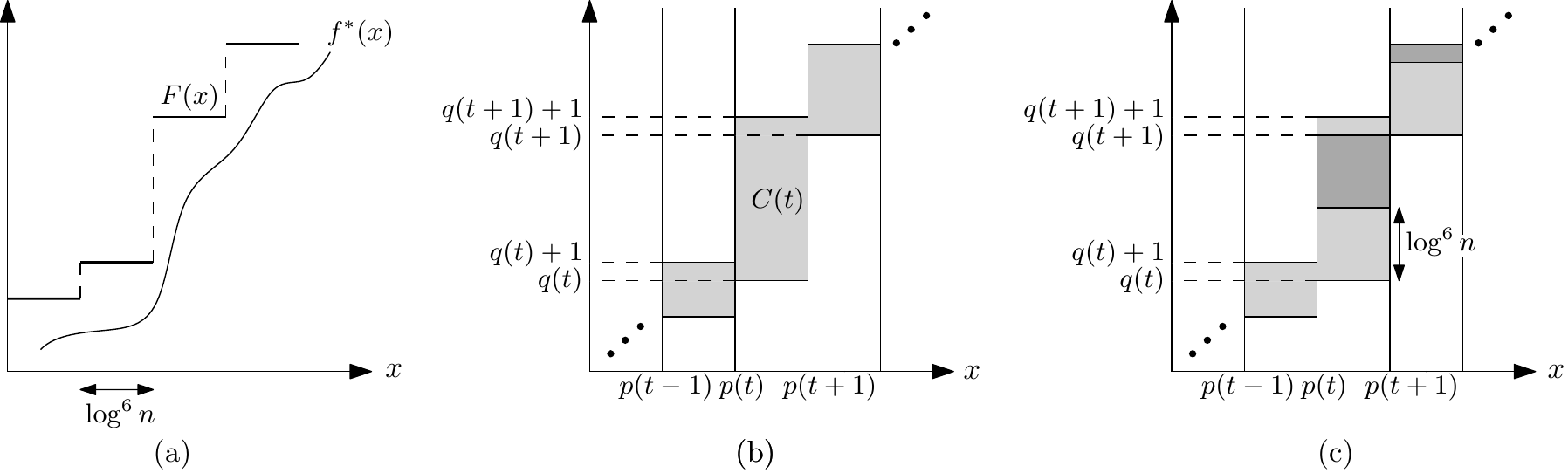}}
      \caption{The scale of the axes are labeled with the indices of vertices. (a) Step function $F(x)$ satisfying $f^*(x)\leq F(x)$,
      with intervals, each consisting of $\log^6 n$ consecutive vertices. (b) Cell $C(t):=[p(t),p(t+1)]\times[q(t),q(t+1)+1]$. 
      (c) $C(t)$ is partitioned further into three subcells, each corresponding to a group
      consisting of at most $\log^{6}{n}$ consecutive indices in $[q(t),q(t+1)+1]$.}
      \label{fig:subcell}
    \end{figure}

  In Lemmas~\ref{lem:farthest} and ~\ref{lem:2chain}, we show how to partition $\partial P$ into two subchains.
  For any two points $x,y\in \partial P$, let $\tau(x,y)$ be the smallest value such that 
  $P_{x,y}$ is $\tau(x,y)$-coverable.
  
  Note that $\tau(x,y)$ is a continuous function while $x$ or $y$ moves along the polygon boundary.
  For a point $p\in \partial P$, let $h(p)$ be the farthest point
  from $p$ in counterclockwise direction
  along $\partial P$ that satisfies $\tau(p,h(p)) \leq \tau(h(p),p)$.
  
  \begin{lemma}\label{lem:farthest}
  Given a point $p\in \partial P$, we can find $h(p)$ in $O(n\log n)$ time.
  \end{lemma}
  \begin{proof}
  For any two vertices $u$ and $v$ of $P$, we can compute $\tau(u,v)$ in $O(n)$ time 
  using the minimum enclosing disk algorithm~\cite{MEGI1984}.
  Observe that, $\tau(u,x)$ increases and $\tau(x,u)$ decreases 
  while $x$ moves from $u$ in counterclockwise direction along $\partial P$. 
  Thus, given a point $p\in \partial P$, we can find the edge $v_iv_{i+1}$ that contains $h(p)$ 
  using binary search in $O(n\log n)$ time.
  
  To compute $\tau(p,q)$ for any two points $p,q\in\partial P$, 
  consider the intersection $E(p,q)$ of cones, each cone expressed as 
  $(x-u_x)^2 + (y-u_y)^2 \leq z^2$ for a vertex $u=(u_x,u_y)$ of $P_{p,q}$ in the $xyz$-coordinate system of $\mathbb{R}^3$.
  Since the intersection of cone $(x-u_x)^2 + (y-u_y)^2 \leq z^2$ and $z=r$ is 
  the disk of radius $r$ centered at $(u_x,u_y)$ contained in $z=r$,
  $E(p,q)$ represents $I_r(p,q)$ for all $r$.
  Thus, $\tau(p,q)$ is the $z$-coordinate of the lowest point in $E(p,q)$.
  The intersection of $n$ translates of a cone with apices
  in convex position in the plane has complexity $O(n)$, and it 
  can be computed in $O(n)$ time~\cite{Aggarwal1989}.
  While $q$ moves along $v_iv_{i+1}$, $\tau(p,q)$ can be obtained from the intersection $E(p,v_i)$ and $E(v_i,q)$. 
  Thus $\tau(q,p)$ has complexity $O(n)$ for all $q \in v_iv_{i+1}$ and it can be computed in $O(n)$ time.
  We can find $h(p)$ by comparing $\tau(p,q)$ and $\tau(q,p)$ for all $q$. 
  Since the binary search for finding the edge containing $h(p)$ dominates the running time, it takes $O(n\log n)$ time in total.
  \end{proof}
  

  \begin{lemma}\label{lem:2chain}
    $P_{v_1,h(v_1)}$ and $P_{h(v_1),v_1}$ is a partition of $\partial P$ such that there are 
    $x \in P_{v_1,h(v_1)}$ and $x' \in P_{h(v_1),v_1}$, and $P_{x,x'}$ and $P_{x',x}$ are $r^*$-coverable.
  \end{lemma}
  \begin{proof}
  As a point $p$ moves along $\partial P$ in counterclockwise direction, $h(p)$ also moves monotonically 
  along $\partial P$ in the same direction; otherwise, $h(p)$ does not satisfy the condition 
  that it is the farthest point or $\tau(p,h(p)) \leq \tau(h(p),p)$.
  
  Consider a point $x \in P_{v_1,h(v_1)}$ such that $P_{x,x'}$ and $P_{x',x}$ are $r^*$-coverable for some point $x'$. 
  Then $r^*=\tau(x,h(x))=\tau(h(x),x)$ since $\tau$ is a continuous function.
  By the monotonicity of $h(p)$, $h(v_1)<_x h(x)$.
  If $\tau(h(v_1),v_1) \neq r^*$, $r^*= \tau(x,h(x)) < \tau(h(v_1),v_1)$. 
  Thus $h(x)<_x v_1$, implying $h(x) \in P_{h(v_1),v_1}$. Note that 
  $P_{x,h(x)}$ and $P_{h(x),x}$ are $r^*$-coverable.    
  If $\tau(h(v_1),v_1) =r^*$, $P_{v_1,h(v_1)}$ and $P_{h(v_1),v_1}$ are $r^*$-coverable. 
  \end{proof}
  
  Now we partition $\partial P$ into two subchains $P_{v_1,h(v_1)}$ and $P_{h(v_1),v_1}$. 
  We consider $x$ moving along $P_{v_1,h(v_1)}$ and $f(x)$ moving along $P_{h(v_1),v_1}$.
  Also, from now on, we use $<$ instead of $<_{v_1}$. The same goes for $\leq_{v_1}, >_{v_1}, \geq_{v_1}$.
  We need the following technical lemma to compute ROI. 

    \begin{lemma}[\protect{\cite{CHAN1999,EPPS1997}}]\label{lem:computeAB}
      After $O(n\log{n})$-time preprocessing, given any vertex pair $(v_i,v_j)$
      we can compute $\tau(v_i,v_j)$ in $O(\log^{6}{n})$ time.
    \end{lemma}

    If an optimal pair of disks cover two subchains $P_{u,v}$ and $P_{v,u}$ for vertices $u$ and $v$ of $P$, 
    we can compute $r^*$ by computing the two centers for vertices in $O(n\log n)$ time 
    using the two-center algorithm for points in convex position~\cite{CHOI2020}.
    Thus, for the following lemma, we assume $r^* < \max(\tau(u,v),\tau(v,u))$ for all polygon vertices $u$ and $v$.
    For simplicity, we use $\tau(i,j)$ for $\tau(v_i,v_j)$ for polygon vertices $v_i$ and $v_j$.

    \begin{lemma}\label{lem:preprocessing}
      We can compute $O(n/\log^6 n)$ disjoint cells of height and width at most $O(\log^6 n)$ 
      such that every intersection of $f^*$ and $g^*$ lies in one of the cells in $O(n\log{n})$ time.
    \end{lemma}
    \begin{proof}
      For ease of description, assume that $P_{v_1,h(v_1)}$ has $k=O(n)$ vertices.
      Let $m = \lfloor {k / \log^{6}{n}}\rfloor$, and $p(t) = t\cdot \lfloor{k/m}\rfloor$ for $t = 1,2, ..., m-1$.
      For convenience, let $v_{p(0)} = v_1$ and $v_{p(m)}=h(v_1)$.
    
      For every vertex $v_{p(t)}$ from $v_1$, we find $v_{q(t)}$ such that 
      $f^* (v_{p(t)}) \leq v_{q(t)+1}$ and $g^*(v_{p(t)}) \geq v_{q(t)}$.
      Using the vertices, we define $F(x)=v_{q(t+1)+1}$ for $x$ with $v_{p(t)} \leq x < v_{p(t+1)}$, and 
      $G(x)= v_{q(t)}$ for $x$ with $v_{p(t)} < x \leq v_{p(t+1)}$.
      
      We compute $q(t)$ as follows. Let $r_t$ be the minimum of $\max\{\tau(p(t),i),\tau(i,p(t))\}$ for indices $i$ in $[k+1,n]$.
      Note that no optimal pair of disks cover two subchains $P_{v_{p(t)},v_i}$ and  $P_{v_i, v_{p(t)}}$, 
      due to the assumption that $r^*<\max(\tau(u,v),\tau(v,u))$ for polygon vertices $u$ and $v$.
      Thus, $r^*<r_t$.
      Assume that $r_t=\tau(p(t),i)$. 
      Then the largest index $q(t)$ satisfying $\tau(p(t),q(t)) < r_t$ also satisfies $r^* < \tau(p(t),q(t)+1)$ and $r^* < r_t < \tau(q(t),p(t))$.
      Thus, $f^* (v_{p(t)}) \leq v_{q(t)+1}$ and $g^*(v_{p(t)}) \geq v_{q(t)}$.
      The case that $r_t=\tau(i,p(t))$ can be shown by a similar argument.
     
      By Lemma~\ref{lem:computeAB}, we can find $q(t)$ in $O(\log^{7}{n})$ time using binary search.
      Since $m=\lfloor k/\log^{6}{n} \rfloor$, we can compute all $q(t)$'s in $O(n\log{n})$ time.
      Thus, we get $m$ cells $C(t):=[p(t),p(t+1)]\times[q(t),q(t+1)+1]$ for $t \in [0,m-1]$.
      See Fig.~\ref{fig:subcell}(b).
  
      We partition each cell $C(t)$ further into subcells as follows.
      If $q(t+1) - q(t) + 1 \leq \log^{6}{n}$, we simply form a group consisting of at most 
      $\log^{6}{n}$ consecutive indices of polygon vertices, from $q(t)$ to $q(t+1)$. 
      If $q(t+1) - q(t) + 1 > \log^{6}{n}$, we partition the indices of polygon vertices, from $q(t)$ to $q(t+1)$, into
      disjoint groups consecutively, each consisting of $\log^{6}{n}$ consecutive indices, 
      except for the last group with at most $\log^{6}{n}$ indices. 
      See Fig.~\ref{fig:subcell}(c) for an illustration for grouping consecutive indices.
      In this way, we have at most $O(n/\log^6 n)$ groups in total, each consisting of at most $\log^{6}{n}$ consecutive indices in $[q(t),q(t+1)+1]$ for $t\in[0,m-1]$.
    \end{proof}
    
    Observe that the $O(n/\log^6 n)$ cells of Lemma~\ref{lem:preprocessing} 
    form ROI. We say that a vertex pair $(v_i,v_j)$ is in ROI if and only if 
    $i\in [p(t),p(t+1)]$ and $j\in [q(t),q(t+1)+1]$ for some $t\in [0,m-1]$ where $m=\lfloor n/\log^6 n\rfloor$. 
    Also, we say an edge pair $(v_iv_{i+1},v_jv_{j+1})$ is in ROI if and only if 
    $i, i+1 \in [p(t),p(t+1)]$ and $j, j+1\in [q(t),q(t+1)+1]$ for some $t\in [0,m-1]$.

    \fi

  \subsection{Decision phase}
  Recall that our parallel decision algorithm finds, for a given $r$, 
  the intersections of the graphs of $f(x)$ and $g(x)$ in ROI.
  We use the data structure of the parallel decision algorithm of the two-center problem for points in convex position~\cite{CHOI2020}.
  To evaluate $f(x)$ for a given $r$, 
  we first find $O(n)$ edge pairs $(e(x),e(f(x)))$ in ROI.
  Then we assign a processor to each edge pair to compute the event points. 
  Then we assign a processor to each event point to compute 
  the breakpoints and the corresponding combinatorial structures of $f(x)$.
  We also do this for $g(x)$.
  Lastly, for each combinatorial structure, we determine whether there exists 
  $x\in\partial P$ such that $f(x) \geq g(x)$.
  This process can be done in $O(\log n)$ parallel steps using $O(n)$ processors,
  after $O(n \log n)$-time preprocessing.

  
  \subsubsection{Data structures.} We adopt the data structure for the two-center problem for points 
  in convex position by Choi and Ahn~\cite{CHOI2020}. 
  To construct the data structure, they store the frequently used intersections of disks for all $r>0$. 
  Then, they find a range of radii $(r_1,r_2]$ containing the optimal radius $r'$ for the two center problem for points in convex position. 
  To do this they use binary search and the sequential decision algorithm for points in convex position. 
  In our case, we compute a range of radii $(r_1,r_2]$ containing the optimal radius $r^*$ 
  using binary search and the sequential decision algorithm in Section~\ref{sec:sequential} running in $O(n)$ time.
  For $r \in (r_1,r_2]$, we construct a data structure that supports the following.

    \begin{lemma}[\cite{CHOI2020}]\label{lem:datastructure}
      After $O(n\log n)$-time preprocessing, we can construct a data structure in $O(\log n)$ parallel steps
      using $O(n)$ processors that supports the following queries with $r \in (r_1,r_2]$: (1) For any vertex $v_i$ in $P_{v_1,h(v_1)}$,
      compute $I(v_i,h(v_1))$ represented in a binary search tree with height $O(\log n)$ in $O(\log n)$ time. 
      (2) For any pair $(v_i,v_j)$ of vertices in ROI, determine if $I(v_i,v_j)=\emptyset$ 
      in $O(\log^3 \log n)$ time.
    \end{lemma}
  
   
  \subsubsection{Computing edge pairs.}
    Using the data structure in Lemma~\ref{lem:datastructure}, we get the following lemma. 
    \begin{lemma}\label{lem:ComputeEdgePairs} 
      Given $r \in (r_1,r_2]$, we can compute all edge pairs $(e(x),e(f(x)))$ in ROI 
      in $O(\log{n})$ parallel time using $O(n)$ processors, after $O(n\log n)$-time preprocessing.
    \end{lemma}
    \iffull
    \begin{proof}
    By Lemma~\ref{lem:datastructure}, we can construct a data structure supporting the following query: 
    given any pair of vertices $(v_i,v_j)$ in ROI, determine if $I_r(v_i,v_j)=\emptyset$ 
    in $O(\log^3 \log n)$ time. For each vertex $v_i$, we find a vertex $v_j$ such that 
    $(v_i,v_j)$ is in ROI, $I_r(v_i,v_j)\neq\emptyset$, and $I_r(v_i,v_{j+1})=\emptyset$ 
    in $O(\log^4 \log n)$ time using the data structure and performing binary search on $\log^6 n$ vertices.
    It takes $O(\log n)$ parallel time using $O(n)$ processors to construct the data structure, 
    after $O(n\log n)$-time preprocessing. Thus, we find all edge pairs in ROI in $O(\log n)$ parallel time using $O(n)$ processors, 
    after $O(n\log n)$-time preprocessing.
    \end{proof} 
  \fi
  
  \subsubsection{Computing the combinatorial structure.}
    After computing the edge pairs using Lemma~\ref{lem:ComputeEdgePairs}, we compute 
    the breakpoints and the corresponding combinatorial structures of $f(x)$.
    To do this, we compute event points and find breakpoints from the event points for each edge pair. 
    For $D_1(x)$ of type \texttt{T3} or \texttt{T4},
    its determinators never change for an edge pair $(e(x),e(f(x)))$ 
    \iffull by Lemma~\ref{lem:types34}. \fi 
    Thus, for each edge pair we find candidates of the determinators of $D_1(x)$ of type \texttt{T3} or \texttt{T4} in $O(\log n)$ time.
    
    For $D_1(x)$ of type \texttt{T1} or \texttt{T2}, we find the event points of $\ccw(x)$ of $\alpha(x,h(v_1))$, 
    and the event points of $\cw(f(x))$ of $\alpha(h(v_1),f(x))$.
    Consider an edge pair $(u'u,vv')$ in ROI such that $f(x)\in vv'$  for some $x \in u'u$.
    The edge pair $(u'u,vv')$ may have $O(n)$ event points at which $\ccw(x)$ of $\alpha(x,h(v_1))$ or 
    $\cw(f(x))$ of $\alpha(h(v_1),f(x))$ changes, while the total number of event points is $O(n)$.
    We find the event points of $\ccw(x)$ of $\alpha(x,h(v_1))$ represented in a binary search tree using $I(u,h(v_1))$ in $O(\log n)$ time.
    Thus, we can find the event points of $\ccw(x)$ of $\alpha(x,h(v_1))$ for all edge pairs in $O(\log n)$ parallel steps
  using $O(n)$ processors.
    For two consecutive event points of $\ccw(x)$ of $\alpha(x,h(v_1))$, 
    we compute the corresponding event points of $\cw(f(x))$ of $\alpha(h(v_1),f(x))$ in $O(\log n)$ time. 
  For a segment $T$ such that $e(x)$, $\ccw(x)$ of 
  $\alpha(x,h(v_1))$, $e(f(x))$, and $\cw(f(x))$ of $\alpha(h(v_1),f(x))$ 
  remain the same for any $x\in T$, we compute $f(x)$.
    
  
    \begin{lemma} \label{lem:ParCompute}
      Given $r \in (r_1,r_2]$, we can compute $f(x)$ for all $x\in \partial P$  such that 
      $(e(x),e(f(x)))$ is an edge pair in ROI, represented as a binary search tree 
      of height $O(\log n)$ consisting of $O(n)$ nodes, in $O(\log{n})$ parallel steps 
      using $O(n)$ processors, after $O(n\log n)$-time preprocessing.
    \end{lemma}
    \iffull
    \begin{proof}
    Let $(u'u,vv')$ be an edge pair in ROI such that $f(x)\in vv'$  for some $x \in u'u$. 
    Imagine $x$ moves along the edge $u'u$ from $u'$ to $u$.
    We compute $f(x)$ for each type of $D_1(x)$.
    Consider the case that $D_1(x)$ is of type \texttt{T3} or \texttt{T4}. 
    Let $x'\in u'u$ be the first breakpoint of $f(x)$ from $u'$ such that $D_1(x)$ becomes a disk 
    of type \texttt{T3} or \texttt{T4}.
    Then the determinators of $D_1(x)$ remain the same for any $x\in x'u$ by Lemma~\ref{lem:types34}.
    Thus, $P_{u,f(x)}$ has the same determinators and $I(u,f(x))$ must be a point.
    This implies that the intersection of $I(u,v)$ and the disk of radius $r$ centered at $f(x)$ is just a point. 
    We can find $I(u,v)$ and $f(x)$ in $O(\log n)$ time by Lemma~\ref{lem:datastructure} and binary search. 
    If the intersection point is a vertex $w$ of $I(u,v)$, then $P_{u,f(x)}$ is covered by 
    $D_1(x)$ of type \texttt{T4}.
    Moreover, the two vertices of $P_{u,f(x)}$ defining $w$ are the two determinators of $D_1(x)$ other than $f(x)$.
    See Fig.~\ref{fig:parallelBreakpoints}(a).
    If the intersection point is in an arc of $I(u,v)$, then 
    $P_{u,f(x)}$ is covered by $D_1(x)$ of type \texttt{T3},
    and the other determinator is the vertex of $P_{u,f(x)}$ defining the arc.
  
  For disks $D_1(x)$ of type \texttt{T1} or \texttt{T2}, we find the event points of 
  $e(x)$, $\ccw(x)$ of $\alpha(x,h(p))$, $e(f(x))$, and $\cw(f(x))$ of $\alpha(h(v_1),x)$.
  To find the event points of $e(f(x))$, we first find the maximal segment $ss'$ contained in $u'u$ 
  such that $e(f(x))$ is $vv'$ for any $x \in ss'$.
  Let $D(x)$ be the disk of radius $r$ centered at $x$.
  Note that $s$ in $e(x)$ is the closest point from $u'$ such that $D(s) \cap I(u,v) \neq \emptyset$, 
  and $s'$ in $e(x)$ is the farthest point from $u'$ such that $D(s') \cap I(u,v') \neq \emptyset$. 
  We can find $s$ and $s'$ in $O(\log n)$ time by binary search. See Fig.~\ref{fig:parallelBreakpoints}(b).

    For $x \in ss'$, we find the event points of $\ccw(x)$ of $\alpha(x,h(v_1))$ using $I(u,h(v_1))$.  
    The vertices of $P$ that are the center of the arcs of $\alpha(u,x_1)$ lying in between $\partial D(s) \cap \partial I(u,h(v_1))$
    and $\partial D(s') \cap \partial I(u,h(v_1))$ are $\ccw(x)$ of $\alpha(x,h(v_1))$ while $x$ moves along $ss'$.
    This is due to the duality between $I(x,h(v_1))$ and $\alpha(x,h(v_1))$.
    See Fig.~\ref{fig:parallelBreakpoints}(c) for an illustration.
    
    Similarly, we find the maximal segment $tt'$ such that $e(x)$ is $u'u$ for all $f(x) \in tt'$ and 
    compute $\cw(f(x))$ of $\alpha(h(p),f(x))$ for $f(x) \in tt'$ in $O(\log n)$ time. 
    Thus, we can find the event points contained in all edge pairs in $O(\log n)$ time using $O(n)$ processors. 
    
    Let $z$ and $z'$ be two consecutive event points of $\ccw(x)$ of $\alpha(x,h(v_1))$ such that $zz'$ is contained in $e(x)$.
    We can find $\cw(f(z))$ of $\alpha(h(v_1),f(z))$ and $\cw(f(z'))$ of $\alpha(h(v_1),f(z'))$ in $O(\log n)$ time 
    by binary search over the event points of $\cw(f(x))$ of $\alpha(h(v_1),f(x))$.
    The event points of $\cw(f(x))$ of $\alpha(h(v_1),f(x))$ lying in between $f(z)$ and $f(z')$ are represented 
    as a binary search tree of height $O(\log n)$.
   We repeat this process for every segment in $e(x)$ connecting two consecutive event points of $\ccw(x)$ of $\alpha(x,h(v_1))$. 
   Then we have all event points of $\cw(f(x))$ of $\alpha(h(v_1),f(x))$ lying in between 
   two consecutive event points of $\ccw(x)$ of $\alpha(x,h(v_1))$.
   By Lemma~\ref{lem:precomputeEvent1} and Corollary~\ref{cor:precomputeEvent2}, there are $O(n)$ event points in total.
   
   Let $T$ be a maximal segment such that for any $x\in T$, $e(x)$, $\ccw(x)$ of $\alpha(x,h(v_1))$, $e(f(x))$, and $\cw(f(x))$ of $\alpha(h(v_1),f(x))$ remain the same.
   By comparing two $f(x)$ functions, one for the current type of $D_1(x)$ and one for other types of $D_1(x)$, we can compute in $O(1)$ time the next breakpoint and the new determinators and type of $D_1(x)$ for $x$ right after the breakpoint.
    Therefore, we can compute $f(x)$ for all $x\in \partial P$ such that $(e(x),e(f(x)))$ is an edge pair in ROI, 
    represented as a binary search tree of height $O(\log n)$ consisting of $O(n)$ nodes, in $O(\log{n})$ parallel steps 
    using $O(n)$ processors, after $O(n\log n)$-time preprocessing.
    \end{proof}
    \fi
    \iffull
    \begin{figure}[ht]%
      \centering
      {\includegraphics[width=\textwidth]{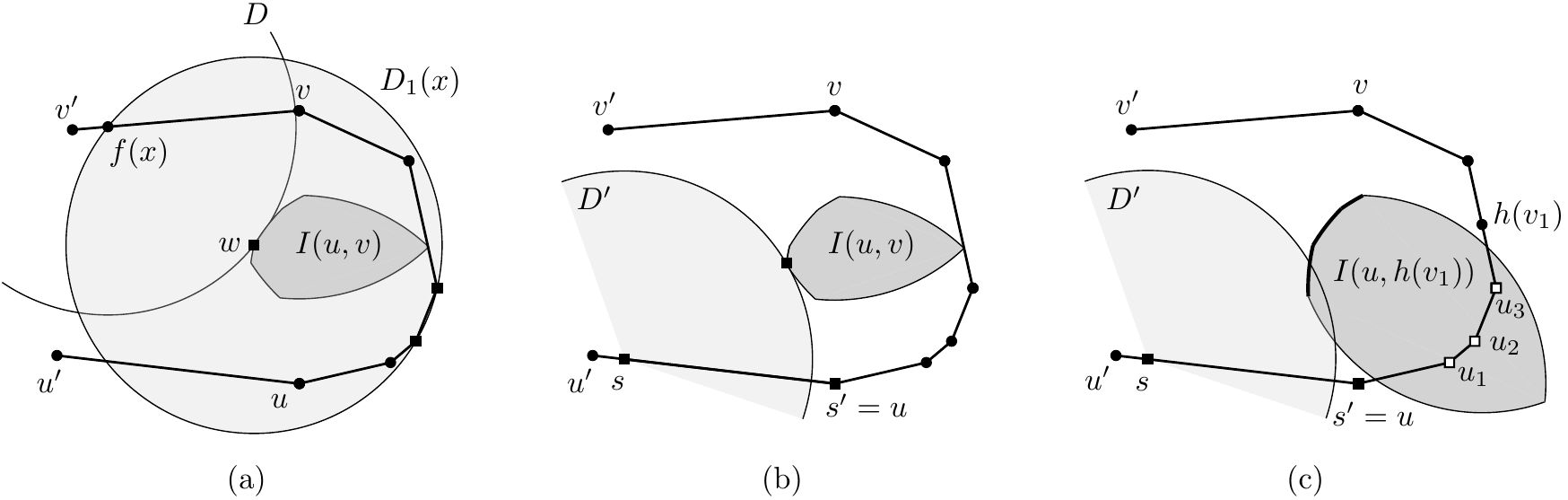}}%
      \caption{Computing the candidates of $D_1(x)$. (a) $D_1(x)$ under the assumption that $D_1(x)$ is of \texttt{T4}. The intersection of $I(u,v)$ and 
      the disk $D$ of radius $r$ centered at $f(x)$ is just a vertex $w$ of $I(u,v)$. 
      (b) Finding a maximal segment $[s,s']$ for $D_1(x)$ of type \texttt{T1}. $D'$ is
      the disk of radius $r$ centered at $s$. (c) Finding $\ccw(x)$'s of $\alpha(x,h(v_1))$ 
      while $x$ moves along $[s,s']$. The bold arcs of $I(u,h(v_1))$ are the arcs defined by $\ccw(x)$'s. $u_1=\ccw(s'), u_3=\ccw(s)$.}%
      \label{fig:parallelBreakpoints}%
    \end{figure}        
    \fi
      
    Now, we have $f(x)$ and $g(x)$ within ROI, each represented as a binary search tree of height $O(\log n)$ and size $O(n)$. 
    For two consecutive breakpoints $t$ and $t'$ of $f(x)$, we find the corresponding combinatorial structures of $g(t)$ and $g(t')$. 
    Then we determine whether there exists $x\in tt'$ such that $f(x) \geq g(x)$ 
    for all combinatorial structures of $g(x)$.
    Since $f(x)$ and $g(x)$ have $O(n)$ breakpoints by Lemma~\ref{lem:SqCompute}, 
    we can determine whether two disks of radius $r$ cover $P$ 
    in $O(\log n)$ parallel steps using $O(n)$ processors, 
    after $O(n\log n)$-time preprocessing.
    Therefore, using Lemma~\ref{lem:ParCompute}, we get the following theorem.
        
    \begin{theorem}
      Given a real value $r$, we can determine whether $r\geq r^*$ in $O(\log{n})$ parallel steps 
      using $O(n)$ processors, after $O(n\log n)$-time preprocessing.
    \end{theorem}
\iffull
By applying Cole's parametric search~\cite{COLE1987} with our sequential decision algorithm and parallel decision algorithm, we get the following theorem.
\fi
    \begin{theorem}  
  Given a convex polygon with $n$ vertices in the plane, we can find in $O(n\log n)$ time 
  two congruent disks of minimum radius whose union covers the polygon. 
  \end{theorem}
  \iffull
  \begin{proof}
    We use Cole's parametric search technique~\cite{COLE1987} to compute the optimal radius $r^*$. 
    For a sequential decision algorithm of running time $T_S$ and 
    a parallel decision algorithm of parallel running time $T_P$ using $N$ processors, 
    we can apply Cole's parametric search and compute $r^*$ in $O(NT_P + T_S(T_P + \log N))$ time.
    To apply Cole's parametric search, the parallel decision algorithm must satisfy 
    a bounded fan-in/bounded fan-out requirement.
    Our parallel decision algorithm satisfies this requirement because it consists of 
    independent binary searches. In our case, $T_S = O(n)$, $T_P = O(\log n)$, and $N = O(n)$.
    Therefore, by applying Cole's technique, $r^*$ can be computed in $O(n \log n)$ time.
  \end{proof}
  \fi
  
  \iffull
  \section{Conclusions}
  We present an $O(n\log n)$-time algorithm for the two-center problem for a convex polygon, 
  where $n$ is the number of vertices of the polygon. However, it is unknown whether the time complexity is optimal. Thus, a question is whether there is any lower bound other than 
  the trivial $\Omega(n)$. We would like to mention that our sequential decision algorithm 
  can be used for covering $\partial P$ 
  (and any curve) using the minimum number of unit disks under the condition that
  a unit disk can cover at most one connected component of $\partial P$.
  \fi

\bibliographystyle{splncs04}
\bibliography{papers}

\begin{thebibliography}{10}
\providecommand{\url}[1]{\texttt{#1}}
\providecommand{\urlprefix}{URL }
\providecommand{\doi}[1]{https://doi.org/#1}

\bibitem{Agarwal2013}
Agarwal, P.K., Ben~Avraham, R., Sharir, M.: The 2-center problem in three
  dimensions. Computational Geometry: Theory and Applications  \textbf{46}(6),
  734--746 (2013)

\bibitem{Agarwal2002}
Agarwal, P.K., Procopiuc, C.M.: Exact and approximation algorithms for
  clustering. Algorithmica  \textbf{33}(2),  201--226 (2002)

\bibitem{AGAR1994}
Agarwal, P.K., Sharir, M.: Planar geometric location problems. Algorithmica
  \textbf{11}(2),  185--195 (1994)

\bibitem{Agarwal1998}
Agarwal, P.K., Sharir, M.: Efficient algorithms for geometric optimization. ACM
  Computing Surveys  \textbf{30}(4),  412--458 (1998)

\bibitem{Agarwal2015}
Agarwal, P., Sharathkumar, R.: Streaming algorithms for extent problems in high
  dimensions. Algorithmica  \textbf{72}(1),  83--98 (2015)

\bibitem{Aggarwal1989}
Aggarwal, A., Guibas, L.J., Saxe, J., Shor, P.W.: A linear-time algorithm for
  computing the {V}oronoi diagram of a convex polygon. Discrete \&
  Computational Geometry  \textbf{4}(6),  591--604 (1989)

\bibitem{Ahn2014107}
Ahn, H.K., Kim, H.S., Kim, S.S., Son, W.: Computing $k$ centers over streaming
  data for small $k$. International Journal of Computational Geometry and
  Applications  \textbf{24}(2),  107--123 (2014)

\bibitem{Ahn2016}
Ahn, H.K., Barba, L., Bose, P., Carufel, J.L.D., Korman, M., Oh, E.: A
  linear-time algorithm for the geodesic center of a simple polygon. Discrete
  \& Computational Geometry  \textbf{56}(4),  836--859 (2016)

\bibitem{AHN2013}
Ahn, H.K., Kim, S.S., Knauer, C., Schlipf, L., Shin, C.S., Vigneron, A.:
  Covering and piercing disks with two centers. Computational Geometry: Theory
  and Applications  \textbf{46}(3),  253--262 (2013)

\bibitem{Bae2019}
Bae, S.W.: ${L}_1$ geodesic farthest neighbors in a simple polygon and related
  problems. Discrete \& Computational Geometry  \textbf{62}(4),  743--774
  (2019)

\bibitem{BASAPPA2020}
Basappa, M., Jallu, R.K., Das, G.K.: Constrained $k$-center problem on a convex
  polygon. International Journal of Foundations of Computer Science
  \textbf{31}(02),  275--291 (2020)

\bibitem{CHAN1999}
Chan, T.M.: More planar two-center algorithms. Computational Geometry: Theory
  and Applications  \textbf{13}(3),  189--198 (1999)

\bibitem{Chan2014}
Chan, T., Pathak, V.: Streaming and dynamic algorithms for minimum enclosing
  balls in high dimensions. Computational Geometry: Theory and Applications
  \textbf{47}(2),  240--247 (2014)

\bibitem{CHO2020}
Cho, K., Oh, E.: Optimal algorithm for the planar two-center problem.
  arXiv:2007.08784  (2020)

\bibitem{CHOI2020}
Choi, J., Ahn, H.K.: Efficient planar two-center algorithms. Computational
  Geometry  \textbf{97},  101768 (2021)

\bibitem{COLE1987}
Cole, R.: Slowing down sorting networks to obtain faster sorting algorithms.
  Journal of the ACM  \textbf{34}(1),  200--208 (1987)

\bibitem{DAS2008}
Das, G.K., Roy, S., Das, S., Nandy, S.C.: Variations of base-station placement
  problem on the boundary of a convex region. International Journal of
  Foundations of Computer Science  \textbf{19}(02),  405--427 (2008)

\bibitem{EDELS1983}
Edelsbrunner, H., Kirkpatrick, D., Seidel, R.: On the shape of a set of points
  in the plane. IEEE Transactions on information theory  \textbf{29}(4),
  551--559 (1983)

\bibitem{EPPS1997}
Eppstein, D.: Faster construction of planar two-centers. In: Proceedings of the
  eighth annual ACM-SIAM Symposium on Discrete Algorithms (SODA 1997). pp.
  131--138 (1997)

\bibitem{Feder1988}
Feder, T., Greene, D.: Optimal algorithms for approximate clustering. In:
  Proceedings of the 20th annual ACM Symposium on Theory of Computing (STOC
  1988). pp. 434--444 (1988)

\bibitem{Fischer2004}
Fischer, K., G{\"a}rtner, B.: The smallest enclosing ball of balls:
  combinatorial structure and algorithms. International Journal of
  Computational Geometry \& Applications  \textbf{4}(5),  341--378 (2004)

\bibitem{GONZALEZ1985}
Gonzalez, T.F.: Clustering to minimize the maximum intercluster distance.
  Theoretical Computer Science  \textbf{38},  293--306 (1985)

\bibitem{Hershberger2008}
Hershberger, J., Suri, S.: Adaptive sampling for geometric problems over data
  streams. Computational Geometry: Theory and Applications  \textbf{39}(3),
  191--208 (2008)

\bibitem{HERSH1991}
Hershberger, J., Suri, S.: Finding tailored partitions. Journal of Algorithms
  \textbf{12}(3),  431--463 (1991)

\bibitem{HWANG1993}
Hwang, R., Lee, R.C.T., Chang, R.: The slab dividing approach to solve the
  euclideanp-center problem. Algorithmica  \textbf{9}(1),  1--22 (1993)

\bibitem{KIM2015}
Kim, S.S., Ahn, H.K.: An improved data stream algorithm for clustering.
  Computational Geometry: Theory and Applications  \textbf{48}(9),  635--645
  (2015)

\bibitem{KIM2000}
Kim, S.K., Shin, C.S.: Efficient algorithms for two-center problems for a
  convex polygon. In: Proceedings of the 6th Annual International Conference on
  Computing and Combinatorics (COCOON 2000). pp. 299--309. Springer (2000)

\bibitem{Loeffler2010}
L{\"o}ffler, M., Van~Kreveld, M.: Largest bounding box, smallest diameter, and
  related problems on imprecise points. Computational Geometry: Theory and
  Applications  \textbf{43}(4),  419--433 (2010)

\bibitem{Megiddo1989}
Megiddo, N.: On the ball spanned by balls. Discrete \& Computational Geometry
  \textbf{4}(1),  605--610 (1989)

\bibitem{MEGI1984}
Megiddo, N.: Linear programming in linear time when the dimension is fixed.
  Journal of the ACM  \textbf{31}(1),  114--127 (1984)

\bibitem{OH2018}
Oh, E., De~Carufel, J.L., Ahn, H.K.: The geodesic 2-center problem in a simple
  polygon. Computational Geometry: Theory and Applications  \textbf{74},
  21--37 (2018)

\bibitem{ROY2008}
Roy, S., Bardhan, D., Das, S.: Base station placement on boundary of a convex
  polygon. Journal of Parallel and Distributed Computing  \textbf{68}(2),
  265--273 (2008)

\bibitem{SADHU2019}
Sadhu, S., Roy, S., Nandi, S., Maheshwari, A., Nandy, S.C.: Two-center of the
  convex hull of a point set: Dynamic model, and restricted streaming model.
  Fundamenta Informaticae  \textbf{164}(1),  119--138 (2019)

\bibitem{SHARIR1997}
Sharir, M.: A near-linear algorithm for the planar 2-center problem. Discrete
  \& Computational Geometry  \textbf{18}(2),  125--134 (1997)

\bibitem{SHIN1998}
Shin, C.S., Kim, J.H., Kim, S.K., Chwa, K.Y.: Two-center problems for a convex
  polygon. In: Proceedings of the 6th annual European Symposium on Algorithms
  (ESA 1998). pp. 199--210. Springer (1998)

\bibitem{WANG2020}
Wang, H.: On the planar two-center problem and circular hulls. In: Proceeding
  of the 36th International Symposium on Computational Geometry (SoCG 2020).
  vol.~164, pp. 68:1--68:14 (2020)

\bibitem{Wang2020full}
Wang, H.: On the planar two-center problem and circular hulls. arXiv preprint
  arXiv:2002.07945  (2020)

\bibitem{Zadeh2008}
Zarrabi-Zadeh, H.: Core-preserving algorithms. In: Proceedings of the 20th
  annual Canadian Conference on Computational Geometry (CCCG 2008). pp.
  159--162 (2008)

\end{thebibliography}
%
\end{document}